\def\year{2018}\relax
\documentclass[letterpaper]{article} 
\usepackage{aaai18}  
\usepackage{times}  
\usepackage{helvet}  
\usepackage{courier}  
\usepackage{url}  
\usepackage{graphicx}  
\usepackage{etoolbox}
\usepackage{blkarray}
\usepackage{algorithm}
\usepackage{algpseudocode}
\usepackage{graphics}
\usepackage{amsthm}

\usepackage{color}
\usepackage{amssymb}
\usepackage{amsfonts}
\usepackage{amsmath,amsthm}

\newtoggle{anon}
\toggletrue{anon}

\newtheorem{example}{Example}

\newcommand{\set}[1]{\{#1\}}			

\newcommand{\E}[1]{E\left[ #1 \right]}
\newcommand{\Em}[2]{\mathbf{E}_{#1} \left[ #2 \right]}
\newcommand{\Pm}{\mathrm{Pr}}         
\newcommand{\Pro}[1]{\Pm \left(#1\right)}

\newcommand{\abs}[1]{\left| #1 \right|}

\newcommand{\bigo}[1]{O\left( #1 \right)}
\newcommand{\bigom}[1]{\Omega \left( #1 \right)}
\newcommand{\bigolog}[1]{\tilde{O} \left( #1 \right)}
\newcommand{\bigomlog}[1]{\tilde{\Omega}\left(#1\right) }

\newtheorem{theorem}{Theorem}
\newtheorem{lemma}[theorem]{Lemma}

\newcommand{\sD}{\mathcal{D}}

\theoremstyle{definition}
\newtheorem{definition}{Definition}[section]
\usepackage[space]{grffile}

\usepackage{array}
\usepackage{graphicx}
\usepackage{epstopdf}
\usepackage{wrapfig}
\usepackage{float}
\usepackage{cleveref}
\usepackage [english]{babel}
\usepackage [autostyle, english = american]{csquotes}
\MakeOuterQuote{"}
\newcommand{\eye}{\mathbb{I}}

\newcount\Comments  
\newcommand{\kibitz}[2]{\ifnum\Comments=1{\color{#1}{#2}}\fi}
\newcommand{\dm}[1]{\kibitz{blue}{[Deb: #1]}}
\newcommand{\dcp}[1]{\kibitz{red}{[Dcp: #1]}}
\newcommand{\todo}[1]{\kibitz{cyan}{[TODO: #1]}}

\begin{document}
%
\title{Peer Prediction with Heterogeneous Tasks}
\newcommand\Mark[1]{\textsuperscript#1}
\author{Debmalya Mandal\Mark{1}, Matthew Leifer\Mark{1}, David C. Parkes\Mark{1}, Galen Pickard\Mark{2} and Victor Shnayder\Mark{1} \\
	\begin{tabular}{cc}
		\Mark{1} Harvard SEAS & \Mark{2} Google \\
		Cambridge, MA 02138 & New York, NY\\
	\end{tabular}\\
	\texttt{\{dmandal@g, matthewleifer@college, parkes@eecs\}.harvard.edu},\\
	\texttt{gpickard@google.com, shnayder@gmail.com}
}
\maketitle
\begin{abstract}
Peer prediction promotes contributions
of useful information by users in settings 
in which there is no way to verify the
quality of responses. This paper introduces
the problem of peer prediction with heterogeneous tasks,
where each task 
is associated with a different distribution on responses.
The  motivation comes from 
eliciting
user-generated content
about places in a city,
where tasks vary because places
and questions about places vary. 
We extend the {\em correlated agreement (CA)}
mechanism~(\cite{shnayder2016informed})
to this setting,
aligning incentives for investing effort
without creating opportunities for coordinated
manipulations. We demonstrate in simulation much better incentive
properties than other mechanisms, 
using 
data from 
user reports on a crowdsourcing platform.
\end{abstract}

\section{Introduction}

Peer prediction refers to the problem of scoring information reports
in settings where the correctness of a report cannot be verified,
either because there is no objectively correct answer or because this
answer is too costly to acquire. This problem arises in diverse
contexts; e.g., peer assessment of assignments in massive open online
courses, and when collecting feedback about a new restaurant.  Peer
prediction algorithms use reports from multiple participants to score
contributions.

Simple approaches compare 
the responses of two users and award them if they agree. But this does
not promote truthful reporting when one user believes that it is
unlikely that another user will have the same opinion.  This problem
can be alleviated by adjusting scores according to the frequency of
reports~\cite{jurca2008truthful,witkowski2012robust,kamble2015truth}.

A limitation of current approaches, however, is that tasks are assumed to be
{\em ex ante} identical, with each task associated with the same 
distribution on reports.
But tasks on various maps platforms, 
which seek to elicit content from users about places in a city,
are quite heterogeneous.  On this kind of platform, a user is
encouraged to answer several 
different types of questions (= tasks) related to the same
place; e.g., ``is the restaurant noisy?,'' ``is it accessible
by wheelchair?,'' or ``does it serve wine?'' The questions
are related to the same place,
yet  the prior beliefs
about the distribution on reports for each type of
question may be  very different.
%

We design a new, multi-task peer prediction mechanism (the {\em
	correlated agreement-heterogeneous mechanism}) that is responsive to
this challenge. This new mechanism shares similar properties with the
earlier {\em correlated agreement} (CA)
mechanism~\cite{shnayder2016informed}.  In particular, it is {\em
	informed truthful} under weak conditions, meaning that it is
strictly beneficial for a user to invest effort and acquire
information, and that truthful reporting is the best strategy when
investing effort, as well as an equilibrium. We demonstrate that the
mechanism has good incentive properties when tested in simulation on
distributions derived from user reports on a popular maps platform.\footnote{Name of platform removed to respect double-blind submission policy. Summary
statistics, that define distributions on pairs of signal reports
and are used for simulations, will be made available.}
%

\subsection{Related Work}

We focus in this brief discussion on mechanisms that are {\em minimal}, in the sense that
they only require signal (or information) reports and do not require
belief reports.  Miller et al.~\citeyear{MRZ2005} introduced the peer
prediction problem and proposed a minimal mechanism that has truthful
reporting in an equilibrium, however the 
mechanism's design requires 
knowledge of the joint 
signal distribution and is 
vulnerable to coordinated misreports.
In response, Jurca and Faltings~\shortcite{jurca2009mechanisms} show
how to eliminate uninformative, pure-strategy equilibria through a three-peer mechanism, and Kong et al.~\shortcite{kong2016putting} provide a method to design robust, single-task,
binary signal mechanisms.

Witkowski and Parkes~\shortcite{witkowski2012robust} first
introduced the combination of learning and peer prediction, coupling the estimation of
the signal prior together with the shadowing mechanism.
There has also been work on making use of reports from a large
population and coupling scoring with estimation.  For a setting with
latent ground truth model, Kamble et al.~\shortcite{kamble2015truth}
provide mechanisms that guarantee strict incentive compatibility with
a large number of agents.  
Radanovic et al.~\shortcite{radanovic2016incentives} provide a mechanism in which
truthfulness is the highest-paying equilibrium in the asymptote of a
large population and with a self-predicting condition that places a
structure on the correlation structure.

Dasgupta and Ghosh~\shortcite{dasgupta2013crowdsourced} show that
robustness to coordinated misreports can be achieved by using reports
across multiple tasks along with access to partial information about
the joint distribution.  The main insight in the DG mechanism is to
reward agents if they provide the same signal on the same task, but
punish them if one agent's report on one task is the same as another's
on another task.  Shnayder et al.~\shortcite{shnayder2016informed}
generalize DG to handle multiple signals, and show how the required
knowledge about the distribution (the correlation structure on pairs
of signals) can be estimated from reports without compromising
incentives.  Their correlated agreement (CA) mechanism rewards pairs
of reports on the same task (penalizes pairs of reports on different
tasks) based on whether signals are positively or negatively
correlated. On the other hand, ~\cite{agarwal2017} generalize the CA
mechanism when users are heterogeneous and derive sample complexity bounds for
learning the reward matrices.
Shnayder et al.~\shortcite{shnayder-ijcai16} adopt replicator dynamics
as a model of population learning in peer prediction, and confirm that
these multi-task mechanisms (including Kamble et
al.~\shortcite{kamble2015truth}) are successful at avoiding uninformed
equilibria. 

To the best of our knowledge, there is no prior work on extending the
design of these multiple-task mechanisms to heterogeneous tasks, where
pairs of reports may be on different types of tasks, with each task
associated with a different signal distribution.

\section{Heterogeneous, Multi-Task Peer Prediction}

Consider two agents, $1$ and $2$, who are members of a large population.
Each agent is assigned to a set of $M = \set{1,2,\ldots,m}$ tasks.
We adopt a binary effort model: if an agent invests effort he incurs a
cost and obtains an informed {\em signal}, otherwise the agent
receives no signal.  There are $n$ signals.  We do not assume that
tasks are {\em ex ante} identical, however, we do assume that the
signals for different tasks are drawn independently.  

Let $S^1_k$ and $S^2_k$ respectively be the signals of agents $1$ and
$2$ for task $k$ (if investing effort). Let $P_k(i,j) = \Pro{S^1_k =
  i, S^2_k = j}$ be the joint probability for a pair of signals
$(i,j)$ on task $k$ and let $P_k(i)$ and $P_k(j)$ be the corresponding
marginal probabilities.  We assume that the agents are exchangeable in
their roles in these distributions, with the same marginal
distributions and joint distributions for any pair of agents.

An agent's strategy maps every task and every received signal to a
reported signal. Agents make reports without knowledge of each others'
reports. We assume that the type of task, and signal about a task
(upon investing effort), is the only information available to an
agent.  For the theoretical analysis, we assume that an agent adopts
the same strategy across all tasks. We leave the analysis of asymmetric strategies for future work.
\footnote{This is without loss
  of generality in the homogeneous task setting
  of~\cite{shnayder2016informed}, but need not be in the present
  context.}
We allow an agent's strategy to be randomized, i.e. a probability
distribution over the set of possible signals. We will write $F$ and
$G$ to denote the mixed strategies of agents $1$ and $2$ respectively.
Let $\eye$ denote the truthful strategy i.e. $\eye(j) = j$. As in~
Shnayder et al.~\shortcite{shnayder2016informed}, we are interested in
the following two incentive properties:
\begin{definition}(Strong Truthful) A peer prediction mechanism is
  {\em strong truthful} iff for all strategies $F,G$ we have
  $E(\eye,\eye) \geq E(F,G)$, where equality may hold only when $F$
  and $G$ are both the same permutation strategy (i.e. a bijection
  from received signals to reported signals.)
\end{definition}
\begin{definition}(Informed Truthful) A peer prediction mechanism is
  {\em informed truthful} iff for all strategies $F, G$ we have
  $E(\eye,\eye) \geq E(F,G)$, where equality may hold only when $F$
  and $G$ are informed strategies (i.e. reports depend on an agent's
  signal).
\end{definition}
These two truthfulness properties imply that truthful reporting is a
strict and weak correlated equilibrium,
respectively~\cite{shnayder2016informed}. They also ensure that there
are no useful, coordinated misreports available to agents.

\subsection{Delta Matrices}

Following~Shnayder et al.~\shortcite{shnayder2016informed} to multiple
types of tasks, a first approach would be to define the following
$n\times n$ matrix for task $k$:
\begin{equation}\label{eq:original-delta}
\Delta_k(i,j) = P_k(i,j) - P_k(i)P_k(j).
\end{equation}

Let $S_k$ be the {\em sign matrix} of $\Delta_k$ i.e. $S_k(i,j) = 1$ if
$\Delta_k(i,j) > 0$ and $S_k(i,j) = 0$ otherwise.  

In the original CA mechanism~\cite{shnayder2016informed}, each task
$k$ is {\em ex ante} identical, and thus has the same delta matrix.
Denote this matrix $\Delta$, with $S$ the corresponding sign matrix.
%
The original
CA mechanism works as follows:
\begin{enumerate}
	\item Let $r^1_k$ ($r^2_k$) be the signal reported by agent $1$ ($2$) on task $k$.
	\item Pick a task $b$ uniformly at random
	as the {\em bonus task}, and
	pick {\em penalty tasks} $l'$ and $l''$ (with $l'\neq l''$)
	uniformly at random from the remaining tasks.
	\item Pay each agent $S(r^1_b, r^2_b) - S(r^1_{l'}, r^2_{l''})$. 
\end{enumerate}

A simple generalization is to pay $S_b(r^1_b, r^2_b) - S_b(r^1_{l'},
r^2_{l''})$, where $S_b$ is the sign matrix corresponding to the bonus
task. But this is not informed truthful for heterogeneous tasks. This
is demonstrated in Example~\ref{ex:dp1}.
\begin{example}[CA is not informed truthful with heterogeneous tasks]
\label{ex:dp1}
	Consider three tasks (1, 2 and 3) 
	with the following joint probability distributions 
	\begin{equation*}
	\begin{array}{cc@{}c@{}c}
	& \begin{array}{cc} Y & N \end{array}
	&	\begin{array}{cc} Y & N \end{array}
	& \begin{array}{cc} Y & N \end{array}\\
	\begin{array}{c}
	Y \\ N
	\end{array}
	&	 \left[\begin{array}{cc}
	0.4 & 0.22 \\
	0.22 & 0.16 \\
	\end{array}\right]
	&	\left[\begin{array}{cc}
	0.7 & 0.14 \\
	0.14 & 0.02 \\
	\end{array}\right]
	& \left[\begin{array}{cc}
	0.4 & 0.22 \\
	0.22 & 0.16 \\
	\end{array}\right] \\
	& (P_1) & (P_2) & (P_3)
	\end{array}
	\end{equation*}
	
	and the following sign matrices:
	\begin{equation*}
	\mathit{sign}(\Delta_1) : \begin{bmatrix}
	1 & 0 \\ 0 & 1
	\end{bmatrix}
	\ \mathit{sign}(\Delta_2) : \begin{bmatrix}
	0 & 1 \\ 1 & 0
	\end{bmatrix}
	\ \mathit{sign}(\Delta_3) : \begin{bmatrix}
	1 & 0 \\ 0 & 1
	\end{bmatrix}
	\end{equation*}

	Suppose each agent adopts the truthful strategy, and
	task 1 is the bonus task, and 
	2 and 3 are the penalty tasks
	for agents 1 and 2,  respectively. Then the
	expected score is 
	\begin{align*}
	&\sum_{i,j} P_1(i,j) S_1(i,j) - P_2(i) P_3(j) S_1(i,j),
	\end{align*}
which evaluates to
	$-0.0216$.  This is true irrespective of whether the penalty tasks for
	1 and 2, respectively, are 2 and 3 or 3 and 2.  Similarly, we can show
	that the expected scores are $-0.1912$ and $-0.0216$ when the bonus
	task is task $2$ and $3$, respectively.
	
	Now consider the case when the first agent always 
	reports $N$. Suppose task $1$ is the bonus task and 
	tasks $2$ and $3$ are the penalty tasks for 1 and 2, respectively.
	The expected score 
	is
	\begin{align*}
	&\sum_{i,j} P_1(i,j) S_1(N,j) - P_2(i)P_3(j) S_1(N,j), 
	\end{align*}
which evaluates to $0$. Similarly, for task 3 and 2 as
	the penalty for 1 and 2, respectively, the expected score is $0.22$.
	So on average, the expected score for task 1 as bonus is $0.11$.
	Similar calculations show expected scores of 0.22 and 0.11, for tasks
	2 and 3 as bonus, respectively. Thus, the CA mechanism fails to be
	informed truthful for this example.
\end{example}

\section{The Correlated-Agreement Heterogeneous (CAH) Mechanism}

In this section, we extend the CA mechanism to handle heterogeneous
tasks. The main idea is to modify the delta matrix for a bonus
task to allow for the implied product distribution on 
signals on penalty tasks.~\Cref{algo:cah} describes the CAH mechanism.
%
\begin{algorithm}
	\caption{CAH mechanism\label{algo:cah}}
	\begin{algorithmic}[1]
		\Require Joint probability distribution $P_b(\cdot,\cdot)$, marginal probability distributions $\{P_l(\cdot)\}_{l\neq b}$ and reports $\{r^1_k,r^2_k\}_{k=1}^m$
		\State $b \leftarrow$ uniformly at random from $\{1,\ldots,m\}$ (bonus task)
		\State $l' \leftarrow$ uniformly at random from $ \{1,\ldots,m\} \setminus \{b\}$ (penalty task assigned to agent 1)
		\State $l'' \leftarrow$ uniformly at random from $ \{1,\ldots,m\} \setminus \{b,l'\}$ (penalty task assigned to agent 2)
		\State Define $\Delta_b(i,j)$ as \begin{equation}\label{eq:defn-delta}
		\!\!\!\!P_b(i,j) - \frac{1}{(m-1)(m-2)}\!\sum_{\stackrel{t',t'' \in [m] \setminus \set{b}}{\& \ t' \neq t''} } \!\!\!\!P_{t'}(i)P_{t''}(j) 
		\end{equation} 
		\State Let $S_b(i,j)$ be the corresponding score matrix i.e. 
		\[S_b(i,j) = 1 \text{ if } \Delta_b(i,j) > 0 \text{ and } S_b(i,j) = 0 \ \text{otherwise}\]
		\State Make payment $S_b(r^1_b,r^2_b) - S_b(r^1_{l'}, r^2_{l''})$
		to each agent
	\end{algorithmic}
\end{algorithm}



In analyzing the properties of CAH, we note that it is sufficient to
consider only deterministic strategies. The proof of this statement is
analogous to Lemma 3.2~\cite{shnayder2016informed}, and uses the fact
that the maximization of a linear function over a convex region is
extremal.

Given this, let $F_i$ ($G_j$) denote the report of agents 1 (2) on
signal $i$ ($j$).
%
The expected score for strategies $F$ and $G$,
conditioned on some 
bonus task $b$, denoted as $E_b(F,G)$, is:
\begin{align}
& \Em{l',l''}{\sum_{i,j} P_b(i,j) S_b(F_i,G_j) - \sum_{i,j} P_{l'}(i) P_{l''}(j) } \notag \\
&=  \sum_{i,j} P_b(i,j) S_b(F_i,G_j)  \notag \\
&- \sum_{\stackrel{l',l'' \in [m] \setminus \set{b}}{\& \ l' \neq l''}} \frac{1}{(m-1)(m-2)} \sum_{i,j} P_{l'}(i)P_{l''}(j)S_b(F_i,G_j)\notag\\
&=  \sum_{i,j} \Delta_b(i,j) S_b(F_i,G_j), \label{eq:dcp1}
\end{align}

where $\ell$ and $\ell''$ denote agent 1 and agent 2's penalty tasks,
respectively.
%
%
Thus, the expected score, averaged over the $m$
possible bonus tasks, is 
\begin{align} \label{CA-score}
E(F,G) = \frac{1}{m} \sum_{b=1}^m E_b(F,G) = \frac{1}{m} \sum_{b=1}^m \sum_{i,j} \Delta_b(i,j) S_b(F_i,G_j)
\end{align}

We now state a property about the 
delta matrices~\eqref{eq:defn-delta}.
%
\begin{lemma}\label{lem:delta-zero}
	For each task $b$, we have $\sum_{i,j} \Delta_b(i,j) = 0$
\end{lemma}

\begin{proof}
	See Appendix \ref{app:defn-zero}
\end{proof}

\subsection{Informed Truthfulness}

The CAH mechanism is informed truthful under a  weak condition on
the signal distributions.
%
%
\begin{theorem}\label{thm:informed-ca3}
	If for each task $b$, $\Delta_b$ is symmetric and each entry of $\Delta_b$ is non-zero, then the CAH mechanism is informed truthful. 
\end{theorem}
\begin{proof}
	For any bonus task $b$, the truthful strategy $(\eye,\eye)$ has higher expected score than any other pair of strategies $F,G$:
	\begin{align*}
	&E_b(\eye,\eye) =  \sum_{i,j} \Delta_b(i,j) S_b(i,j) =  \sum_{i}\sum_j \max(0,\Delta_b(i,j))
	\\& \geq  \sum_{i,j} \Delta_b(i,j) S_b(F_i,G_j) = E(F,G).
	\end{align*}
	
	Consider an uninformed strategy $F$, with $F_i = r$ for all $i$. Then for any $G$,
	the expected score is
	\begin{align*}
	&\sum_{i=1}^n \sum_{j=1}^n \Delta_b(i,j) S_b(r,G_j) =  \sum_{j=1}^n S_b(r,G_j) \sum_{i=1}^n \Delta_b(i,j) \\ &\leq  \sum_{j=1}^n \max(0, \sum_{i=1}^n \Delta_b(i,j)).
	\end{align*}
	
	We need to show the following:
	\begin{align*}
	& \sum_{j=1}^n \max(0, \sum_{i=1}^n \Delta_b(i,j)) < 
	\sum_{j=1}^n \sum_{i=1}^n \max(0,\Delta_b(i,j)).
	\end{align*}
	
	It is enough to show that for each $b$, there exists a column $j$ and two different rows $i_1,i_2$ such that $\Delta_b(i_1,j) > 0$ and $\Delta_b(i_2,j) < 0$.
	Suppose not. Then each column of $\Delta_b$ has either all
        positive entries or all negative entries. Since each entry of
        $\Delta_b$ is non-zero and Lemma~\ref{lem:delta-zero} holds,
        there exist two columns $j_1$ and $j_2$ such that all entries
        of $j_1$ ($j_2$) are positive (negative). This implies
        $\Delta_b(j_2,j_1) > 0$ and $\Delta_b(j_1,j_2) < 0$, which
        contradicts the fact that $\Delta_b$ is symmetric.
\end{proof}

\subsection{Strong Truthfulness}


We state a sufficient condition for the CAH mechanism to satisfy the property
of strong truthfulness.

{\bfseries Condition 1} :
\begin{enumerate}
	\item  $\Delta_b(i,i) > 0$, $\quad \forall b\ \forall i$.
	\item $\sum_{b=1}^m \Delta_b(i,j) < 0$, $\quad \forall i \neq j$.
\end{enumerate}
\begin{theorem}\label{thm:strong-cah}
	If $\{\Delta_b\}_{b=1}^m$ satisfy  Condition 1, 
	then the CAH mechanism is strongly truthful.
\end{theorem}
\begin{proof}
See Appendix \ref{app:strong-cah}.
\end{proof}



Condition 1 is slightly weaker than the {\em categorical}
condition~\cite{shnayder2016informed}.
$\Delta_b$ is categorical if
(1) $\Delta_b(i,i) > 0$ for all signals $i$, and (2) $\Delta_b(i,j) <
0$ whenever $i\neq j$; i.e., same-signal positive
correlation and other-signal negative correlation.  
Condition 1 does not require every
off-diagonal entry to be negative for all tasks $b$, but
only that the average of
the off-diagonal entries is negative.
Categorical and Condition~1 are equivalent when there are only two signals.
%
%
%
%

\subsection{Combining CAH with Estimation}

As with the CA mechanism~\cite{shnayder2016informed}, the CAH
mechanism remains (approximately) informed truthful even when the
statistics used to determine scores are estimated from the reports of
strategic agents.  The reason is that the score matrix that
corresponds to the correct statistics is the best possible score
matrix for agents, and thus they cannot do better by cooperating in
designing an alternate matrix. 

(Algorithm~\ref{algo:learning-cah}) presents the detail-free version of CAH mechanism, which learns the 
delta matrices from the agents' reports. We will refer to this implementation as CAHR (in short for CAH recomputed).
The next theorem proves that CAHR
is $(\varepsilon,\delta)$-informed
truthful.
\begin{theorem}\label{thm:detail-free-cah}
	If there are at least $q = \bigom{ \frac{n}{\varepsilon^2} \log \left( \frac{m}{\delta}\right)}$ agents reviewing each task, for $m$ tasks and $n$ possible signals,
	then with probability at least $1 - \delta$, then CAHR satisfies
	\[
	\E{\eye,\eye} \geq \E{F,G} - \varepsilon \quad \forall F,G 
	\]
\end{theorem}
\begin{proof}
See Appendix \ref{app:detail-free-cah}
\end{proof}
%


Theorem~\ref{thm:detail-free-cah} implies that truthful reporting is
an approximate equilibrium for the detail-free CAH, and that (up to
$\epsilon$) there is no useful joint deviation.
The proof follows from  the fact that any joint
distribution $P_b(\cdot,\cdot)$ (resp. marginal distribution
$P_b(\cdot)$) can be learned with $\bigolog{n^2/\varepsilon^2}$ (resp.
$\bigolog{n/\varepsilon^2}$) samples\footnote{$\bigolog{\cdot}$ is
  $\bigo{\cdot}$ without all the log factors} and observing that $q$
samples from a task gives us $q^2$ samples from the corresponding
joint distribution.
In addition, we can show a general version of~\Cref{thm:detail-free-cah}. Suppose there
are $t$ distinct types of tasks, and  the number of tasks of type $k$ is $m_k$. Then it is sufficient to have $q = \bigomlog{\frac{1}{\sqrt{m_k}}\frac{n}{\varepsilon^2}}$ samples from each task of type $k$. This follows from the observation that if we have
at least $q$ samples from each task of type $k$ then the total number of samples from the joint distribution $P_k(\cdot,\cdot)$ is at least $m_k q^2 = \bigomlog{{n^2}/{\varepsilon^2}}$.
%
%
\begin{algorithm}[h]
	\caption{ 
CAHR mechanism\label{algo:learning-cah}}
	\begin{algorithmic}[1]
		\Require Agent $p$ of a population of $q$ 
		agents provides reviews $(r^p_1,\ldots,r^p_m)$ on
		each of the $m$ tasks.
		\State $T_k(i,j) \leftarrow $ observed freq of signal pair $i,j$ on task $k$.
		\State Pair up the agents uniformly at random, and  run CAH for each pair with the estimated distribution $\{T_k(\cdot,\cdot)\}_{k=1}^m$
	\end{algorithmic}
\end{algorithm}

\subsection{Cross Correlated Agreement}


So far we have assumed that the probabilities of observing signals are
independent across different tasks. However, two users' responses to two
different tasks may be correlated (e.g. consider two questions -- (1) does
this restaurant serve alcohol and (2) does this restaurant serve wine?)
We will write $P_{l',l''}(i,j)$ to denote the
probability that a user sees signal $i$ on task $l'$ and another user
observes signal $j$ on task $l''$. When there are no correlations
among signals for different questions we have $P_{l',l''}(i,j) =
P_{l'}(i) P_{l''}(j)$. 
The {\em Cross
	Correlated Agreement for Heterogeneous Tasks} ({CCAH})
mechanism generalizes { CAH} by using the probabilities
$P_{l',l''}(\cdot,\cdot)$ for different pairs of tasks $(l',l'')$.

\begin{itemize}
	\item {CCAH} is same as {CAH} except it defines $\Delta_b(i,j)$, the $(i,j)$-th entry of delta matrix for task $b$ as : 
	\begin{align*}
	&P_b(i,j) - \frac{1}{(m-1)(m-2)}\sum_{\stackrel{t', t'' \in [m] \setminus \set{b}}{
			\& \ t' \neq t''}}P_{t',t''}(i,j)
	\end{align*}
\end{itemize}


CCAH is strong truthful and informed truthful under similar conditions stated in theorems \ref{thm:strong-cah} and 
\ref{thm:informed-ca3} respectively. Moreover, a sample complexity result analogous to~\Cref{thm:detail-free-cah} holds for a
detail-free implementation of CCAH. This is because if we have at least $q$ samples from both $l'$ and $l''$, then  we have at least $q^2$ samples
from the joint distribution $P_{l',l''}(\cdot,\cdot)$.  

\section{Experimental Results}

\iftoggle{anon}{XYZ}{Google Local Guides} is 
a platform for collecting user generated content in regard to places
on \iftoggle{anon}{a mapping platform}{Google Maps}. A user can provide information by answering `yes',
`no,' or `not sure' to a series of questions.\footnote{We 
	ignore the `not sure' response for a question
	because of unclear semantics: does it mean the user has
	missing information, or the question is not relevant to
	the location. Thus, {\em a priori} it is unclear whether to expect
	correlation between different reports.}
A user is awarded
one point for each contribution,
where a contribution can be a review or a 
photograph  or any update about the place,
 with a maximum of five
points per place. Based on the number of points received a user 
is in  one of five 
levels on the platform, with higher levels providing
better benefits such as free \iftoggle{anon}{online storage}{Google Drive space}, visibility on the
\iftoggle{anon}{XYZ channel}{Local Guides channel}, and access to new 
products before they are generally released.
%
%
%

A type of task is 
specified by a triple of the form:
\begin{align*}
\mathit{Region} \times \mathit{Business Type}
\times \mathit{Question}
\end{align*}

A region is a US state, there are four business types such as
``restaurant,'' ``bar,'' ``public location'' or ``cafe'' (these are
anonymized in our data), and there are 143 distinct questions 
in the data.  The questions are also anonymized,
but
categorized by \iftoggle{anon}{XYZ}{Google} as ``subjective''
or ``factual'' (e.g., ``is this restaurant noisy?'' vs
``does this cafe have free WiFi?'').
%
Each task type has a corresponding
pairwise signal distribution. 

The data are counts of pairs of signal reports, 
broken down by (region, business type, question).
%
The number of different questions (and thus types of tasks) per pair
of region and business type varies from 75 to 135, with an average of
102. There are 51 regions and 4 business types per state. Thus, the
total number of task types for which we have data is around 20,885.

 For the purpose of our simulations we treat the
distributions for these task types as describing the true signal
distributions.  The goal of the experiments is to compare, under this
assumption, the robustness of the CAH mechanism with other 
mechanisms in the literature.
For this, we consider the {\em robust peer truth serum} (RPTS)
mechanism~\cite{radanovic2016incentives} (which sets a score of
$1/P(i)$ for agreement on signal $i$ and $1$ otherwise)\footnote{This
  is equivalent to a scaled version of the rule given by Radanovic et
  al.~\shortcite{radanovic2016incentives}, who prescribe a score of
  $\alpha(1/P(i) - 1)$ for agreement on signal $i$ and $0$ otherwise.
  Thus, our version has equivalent incentive properties in
  expectation.} and the {\em Kamble}~\shortcite{kamble2015truth}
mechanism (which sets a score of $1/\sqrt{P(i,i)}$ for agreement on
signal $i$).

In simulating CAH, we first compute the delta matrices for each task
type using~\Cref{eq:defn-delta}.  For this, we assume for a given
(region, business type, question) that the penalty tasks are sampled
from other questions associated with the
same  (region, business type).
From these delta matrices, we then use~\Cref{eq:dcp1} to compute the
expected score for each question, before averaging these scores over
all questions associated with a (region, business type)
pair. 
%
For the single task, RPTS and Kamble mechanisms, we compute
the score for a (region, business type) by averaging the
individual scores recevied on each question associated with the
(region, business type) pair. Finally, since the payments of CAH are
bounded between $0$ and $1$, we normalize the payments of RPTS and
Kamble to $[0,1]$.\footnote{
For two signals $0$ and $1$, with
  estimated probabilities $\hat{P}(0) > \hat{P}(1)$, RPTS pays more on
  signal $1$ than $0$. Because of this, we set the payment to $1$ for
  agreement on signal $1$, and divide all payments through by the
  unnormalized payment for agreement on signal $1$ (i.e.,
  $1/\hat{P}(1)$.) The rewards of Kamble are normalized analogously. Our 
normalization scheme is static i.e. the normalization constants are not 
recomputed when the probabilities are estimated based on some possible misreports of the agents.}
Along with CAH, we also evaluate CAHR, the empirical version of the CAH mechanism. CAH has access
to the true delta matrices, whereas, CAHR computes the delta matrices based on the reports of the agents
and then uses these delta matrices to score reports.

\subsection{Unilateral Incentives for Truthful Reports}

We consider three kinds of strategic behaviors: {\em constant-0}
(report `yes' all the time), {\em constant-1} (report `no' all the
time) and {\em random} (report `yes' w.p. 0.5). 
%

We first consider unilateral incentives to make truthful reports, for
various assumptions about how the behavior of the rest of the
population. 
As an illustration,  Figure~\ref{fig:1} shows the expected benefit to being
truthful vs following some other behavior, considering the
average score for each (region, business type). We consider, in
particular, the benefit to being truthful vs the alternate
behavior 
%
when $p=0.8$ of the population is truthful and the rest follow
the same, alternate strategy. 
%
%
%
This models 20\% of the agents being able to coordinate on a deviation
from truthful play.  
\footnote{For CAHR, we first recompute the joint probabilities when
  $p$ fraction of the population is truthful and $1-p$ fraction adopts
  some other strategy, and then compute the delta matrices with
  respect to the new joint probability distributions. On the other
  hand, CAH uses the delta matrices computed using the original joint
  probability distributions.}

\begin{figure*}[t!]
	\minipage{0.245\textwidth}
	\includegraphics[width=\linewidth]{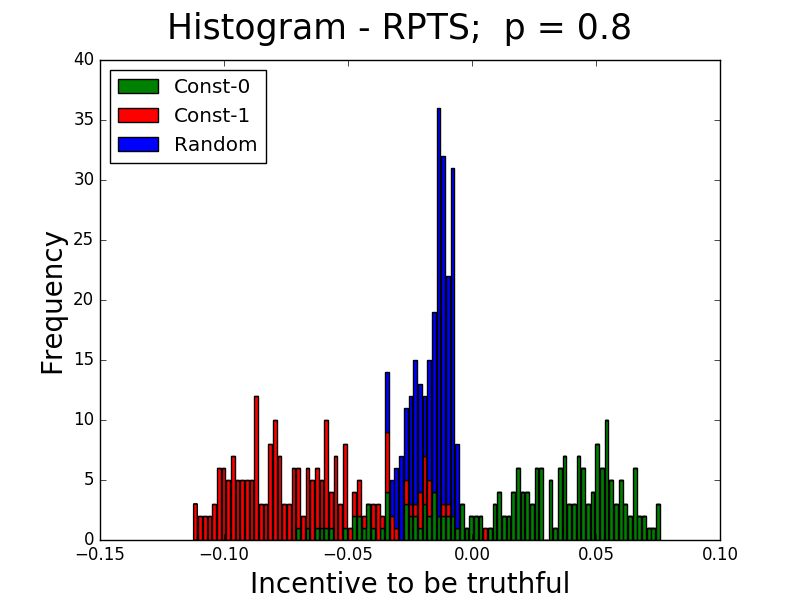}
	\label{fig:etp-rpts}
	\endminipage\hfill
	\minipage{0.245\textwidth}
	\includegraphics[width=\linewidth]{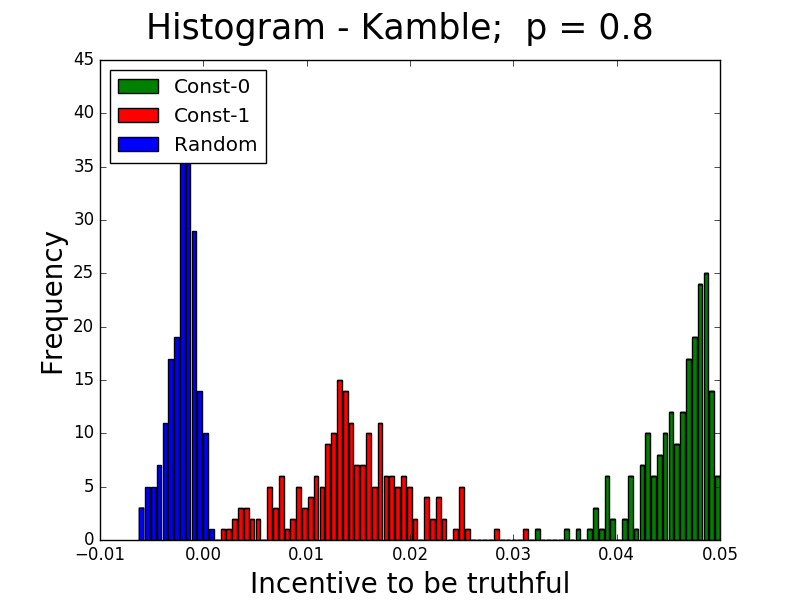}
	\label{fig:etp-kamble}
	\endminipage\hfill
	\minipage{0.245\textwidth}%
	\includegraphics[width=\linewidth]{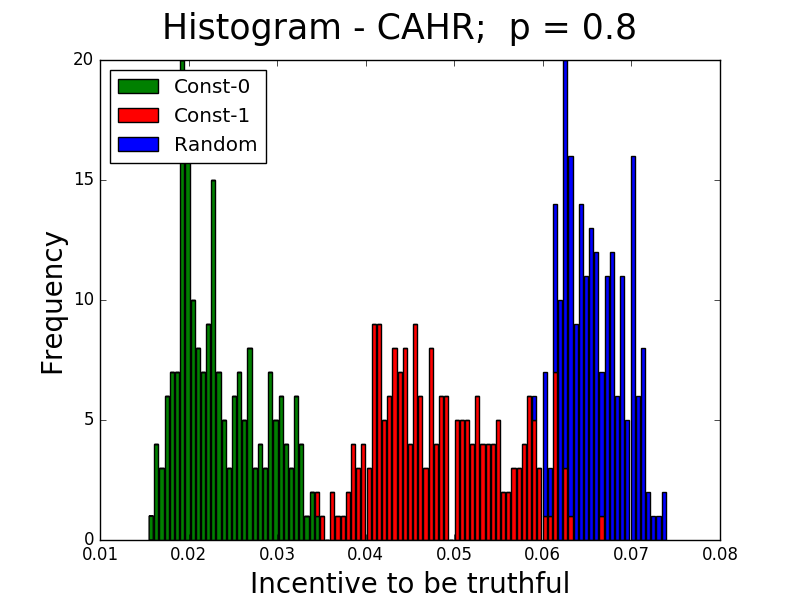}
	\label{fig:etp-cah}
	\endminipage\hfill
	\minipage{0.245\textwidth}
	\includegraphics[width=\linewidth]{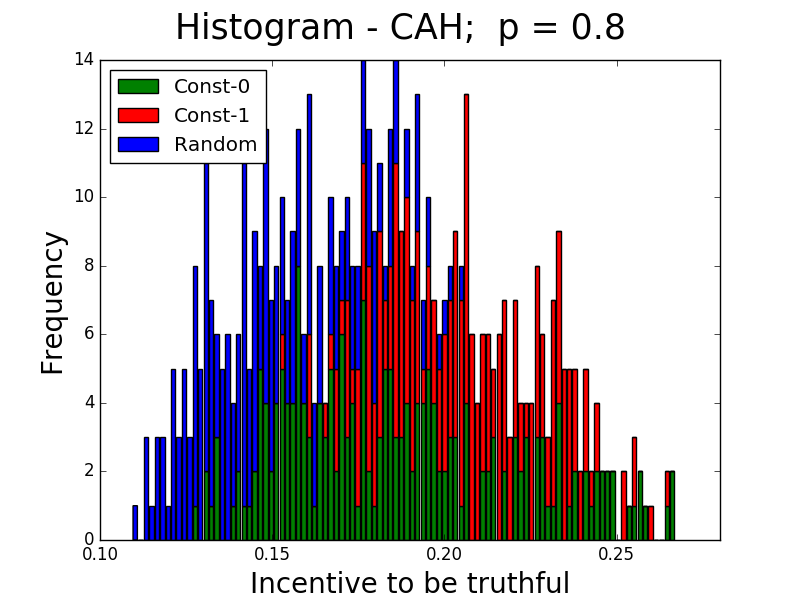}
	\label{fig:etp-cahr}
	\endminipage\hfill
	\caption{Histograms for the 204
		(region, business type) pairs of
		expected benefit
		(averaged across questions) from truthful behavior
		vs.  some other strategy, when fraction $0.8$
		is truthful and fraction $0.2$ adopt the same,
		non-truthful
		strategy.  \label{fig:1}}
\end{figure*}

We observe that the support of the distribution for the CAH and CAHR
mechanism is positive, and thus it retains an incentive for truthful
behavior. We found this to be a common property for different values
of $p$, i.e. CAH and CAHR retains good unilateral incentives for all
values of $p$, even when all agents play the same way. 
By contrast, both the RPTS and Kamble fail under some strategy, i.e.  there exists
a strategy ({\em random} for Kamble and either {\em random} or {\em
  constant-1} for RPTS) such that playing that strategy is more
beneficial than playing truthful strategy when some fraction plays
this alternate strategy.  Although figure~\ref{fig:1} shows this for
$p = 0.8$, this is representative of other values of $p$. The plots
for several other values of $p$ are included in Appendix \ref{app:uni-plots}.

When the prior probability satisfies the {\em self-predicting}
condition, the RPTS mechanism has truth-telling as a strict
equilibrium 
and the truthful
equilibrium provides at least as high payoff than any other coordinated
equilibrium where all agents report the same. Since, incentive properties
are not proven under RPTS except when the self-predicting condition is
satisfied, we evaluated the RPTS
mechanism by restricting only to questions that satisfy the
self-predicting condition. However, the corresponding plot is similar
to the plot shown in figure~\ref{fig:1}.  To conclude, compared to single task mechanisms
like RPTS and Kamble, CAH mechanisms provide good guarantees against
unilateral deviation.

%

\subsection{Benefit from Coordinated Misreports}

Irrespective of whether or not a coordinated deviation is robust
against agents choosing to make truthful reports instead, we also
consider the expected payoff available to a group of agents who manage
to coordinate on some non-truthful play. 
Figure~\ref{fig:2} plots the average and standard error for the expected payments
associated with the 204 (region, business type) pairs. 
For each strategy and for a particular value of $p$, we plot the 
expected payment and the standard error across the 204 pairs, when $p$ fraction of 
population is truthful and the remaining $1-p$ fraction of the
population adopts the same strategy. The constant line shows the average expected payment
 across all the pairs when everyone is truthful. 
 \begin{figure*}[!h]
 	\minipage{0.245\textwidth}
 	\includegraphics[width=0.9\linewidth]{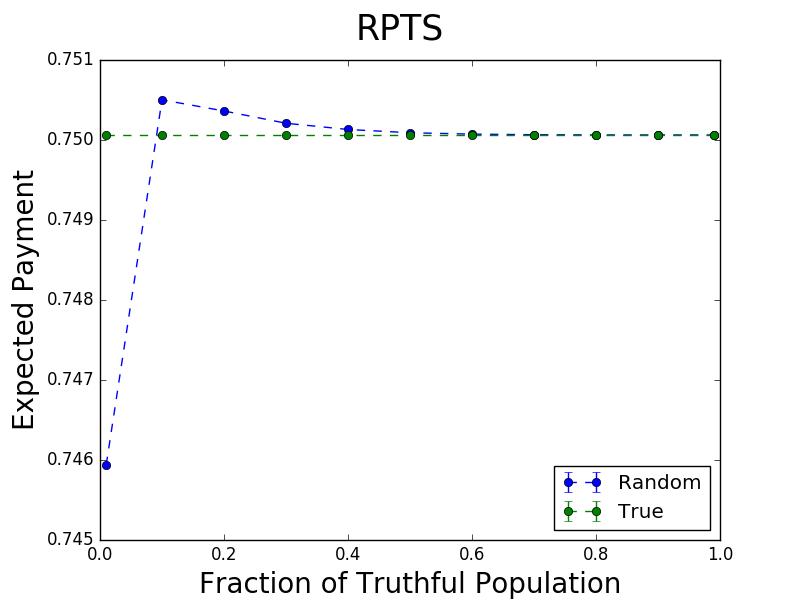}
 	\label{fig:etp_ca}
 	\endminipage\hfill
 	\minipage{0.245\textwidth}
 	\includegraphics[width=0.9\linewidth]{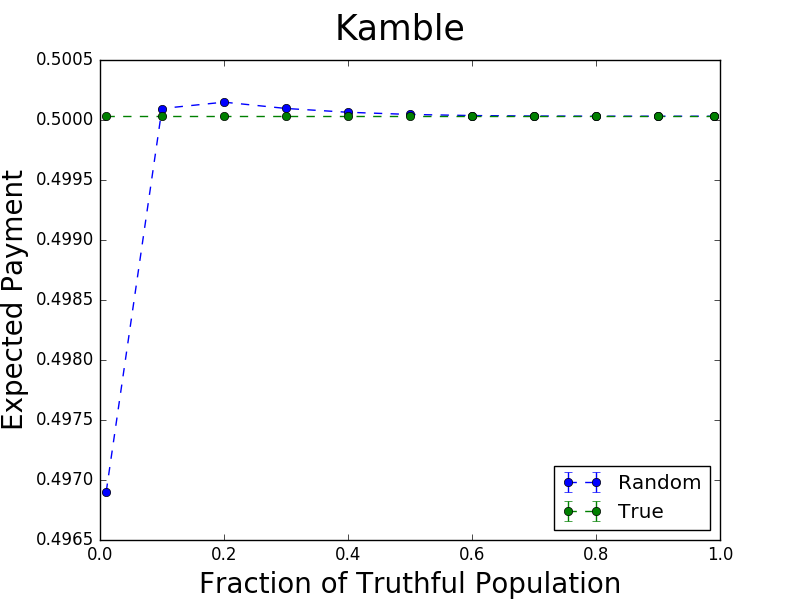}
 	\label{fig:etp-kamble}
 	\endminipage\hfill
 	\minipage{0.245\textwidth}%
 	\includegraphics[width=0.9\linewidth]{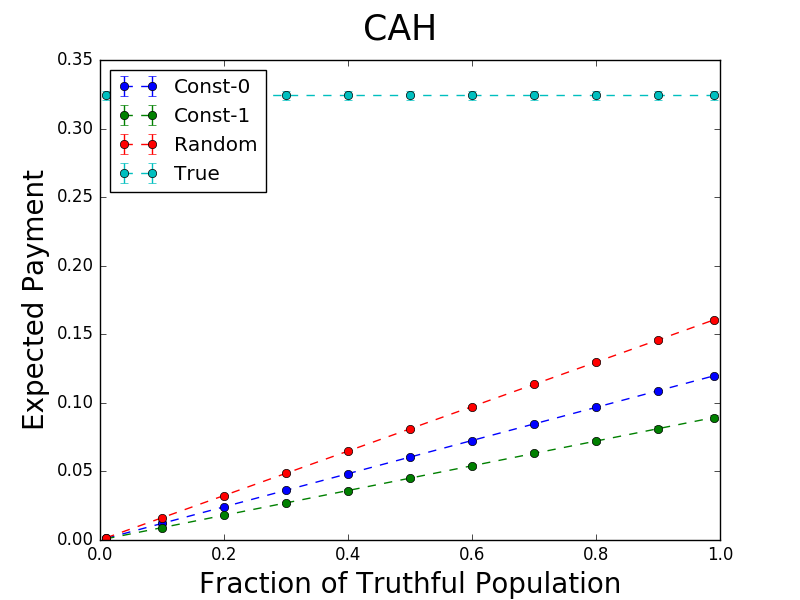}
 	\label{fig:etp-cah}
 	\endminipage\hfill
	\minipage{0.245\textwidth}%
	\includegraphics[width=0.9\linewidth]{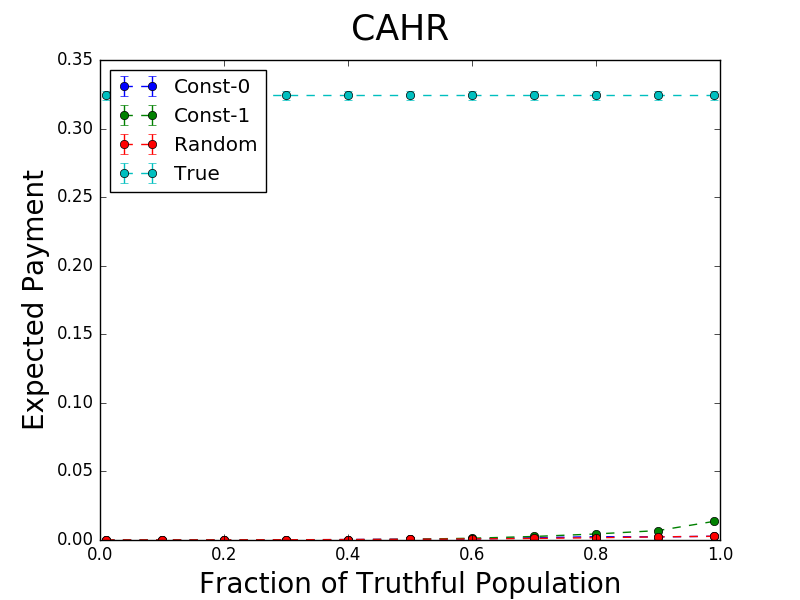}
	\label{fig:etp-cahr}
	\endminipage\hfill
 	\caption{Expected score for following each of four
 		strategies, when
 		$p$ fraction of the population is truthful and  $1-p$ fraction 
 		adopt the same strategy. Averaged over
 		questions associated with a typical (region, business type)
 		pair. \label{fig:2}} 
 \end{figure*}
CAH mechanism has the expected payments from all truthful strategy higher
than the other three possible strategies (const-0, const-1 and random) for all possible
values of $p$. This means that CAH mechanism is robust against coordinated misreport by
any fraction of the population. For RPTS and Kamble, we only plot the expected payments
due to the all truthful strategy and the random strategy for various values for $p$. We
omit the plots for the
expected payments for const-0 and const-1 strategies since the payments under these strategies
are significantly lower than the all truthful
strategy under both RPTS and Kamble mechanism and do not provide profitable coordinated 
misreports. We now see that for intermediate values of $p$, the random strategy provide
a profitable coordinated misreporting profile under both the RPTS and Kamble mechanism.
Therefore, unlike CAH, single task mechanisms like RPTS, Kamble are not always robust
to coordinated deviations.

\subsection{Subjective vs Factual Tasks}

Figure~\ref{fig:3} shows the cumulative
distribution on expected scores at
truthful reporting in each mechanism, where each data point
corresponds to a different (region, business type, question) triple.
Two lines are shown for each mechanism: one corresponding to 
questions that are categorized as `factual' and one
corresponding to questions that are categorized as `subjective.'
\begin{figure*}[t!]
	\minipage{0.33\textwidth}
	\includegraphics[width=0.9\linewidth]{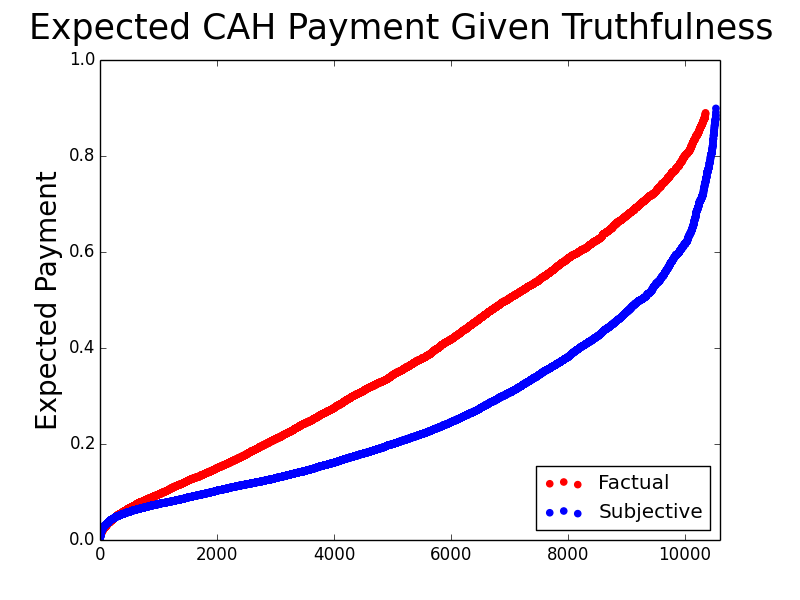}
	\label{fig:etp_ca}
	\endminipage\hfill
	\minipage{0.33\textwidth}
	\includegraphics[width=0.9\linewidth]{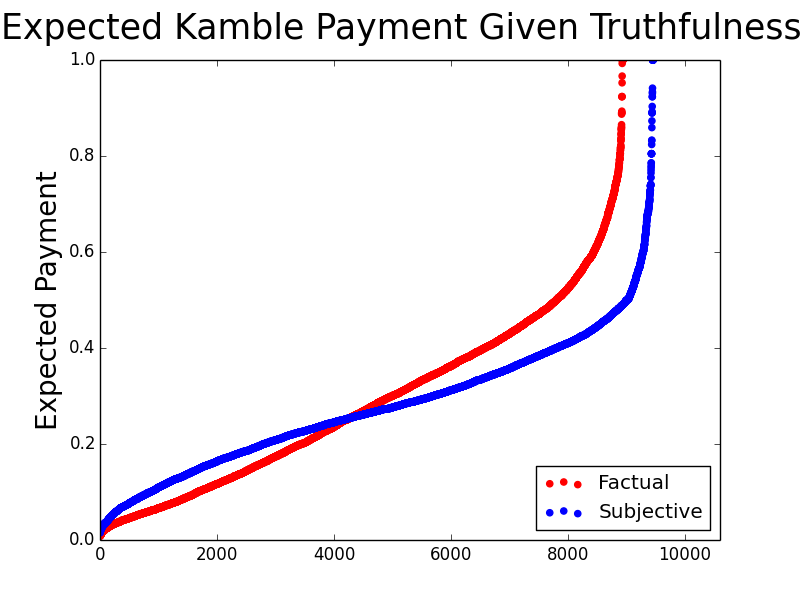}
	\label{fig:etp-kamble}
	\endminipage\hfill
	\minipage{0.33\textwidth}%
	\includegraphics[width=0.9\linewidth]{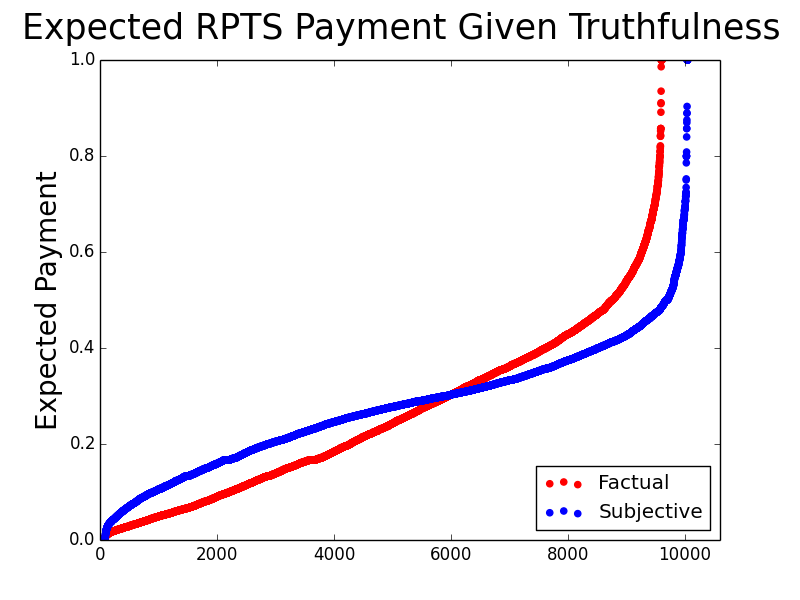}
	\label{fig:etp-rpts}
	\endminipage
	\caption{\label{fig:3} 
		Cumulative distribution on expected payments at truthful reporting
		in each mechanism, with  results separated into questions that are
		categorized as `factual' and those that are categorized as `subjective.'}
\end{figure*}
The subjective
questions tend to provide lower expected payment than the factual
questions under the CAH mechanism. 
%
This is consistent with the intuition
that people perceive subjective questions differently than factual
questions.
For the Kamble and RPTS mechanisms, the variability in 
expected payment is larger across factual questions than 
subjective questions, with the expected payment
for subjective questions tending to fall in a
narrow band.

%

\section{Conclusions}

We study the peer prediction problem when users complete heterogeneous
tasks.  We introduced the CAH mechanism, which is informed-truthful
under mild conditions and can also be used together with estimating
statistics from reports for the purpose of computing scores.  
The simulation results suggest that CAH provides better incentive for
being truthful and is more resistant to coordinated misreports than
the RPTS and Kamble mechanisms. We also noted that CAHR, the empirical version
of CAH has similar incentive guarantees, in contrast to the empirical versions
of the single-task peer prediction mechanisms.
%
%
%
%
We believe that the theoretical guarantees of the multi-task mechanisms and their attractive incentive properties suggest that
such mechanisms are ready to be applied and evaluated in practice, from peer grading to rating.
The most important directions for future work  are to
design mechanisms that can 
handle agent heterogeneity (agents that
vary by taste, judgment, noise, etc.) as
well as task heterogeneity. We are also interested in developing
specific versions of the CAH mechanism for particular models
of heterogeneity, such as the generalized Dawid-Skene scheme~\cite{dawid1979maximum}.

\newpage

\bibliographystyle{named}
\bibliography{references}


\onecolumn

\section{Appendix}
\subsection{Proof of Lemma \ref{lem:delta-zero}}
\label{app:defn-zero}
\begin{align*}
&\sum_{i}\sum_{j} \Delta_b(i,j) = \sum_i \sum_j \left\{P_b(i,j) \right. \\ &-\left.\frac{1}{(m-1)(m-2)} \sum_{l' \in [m] \setminus \set{b}} 
\sum_{l'' \in [m] \setminus \set{b,l'}} P_{l'}(i) P_{l''}(j) \right\} \\
&= \sum_i \left\{ P_b(i) - \frac{1}{(m-1)(m-2)} \sum_{l' \in [m] \setminus \set{b}} 
\sum_{l'' \in [m] \setminus \set{b,l'}} P_{l'}(i)   \right\}\\& = 1-1= 0
\end{align*}

\subsection{Proof of Theorem \ref{thm:strong-cah}}
\label{app:strong-cah}
Suppose both the agents adopt the truthful strategy, which corresponds to the identity matrix $\eye$. Then the expected payment is given as 
\begin{align}
E(\eye,\eye) = \sum_{b=1}^m  \sum_{i,j : \ \Delta_b(i,j) > 0} \Delta_b(i,j)
\end{align}

On the other hand for any two arbitrary deterministic strategies $F$ and $G$, 
\begin{align}
E(F,G) = \sum_{b=1}^m  \sum_{i,j} \Delta_b(i,j) S_b(F_i, G_j) \leq \sum_{b=1}^m  \sum_{i,j : \ \Delta_b(i,j) > 0} \Delta_b(i,j) = E(\eye,\eye)
\end{align}

To show strong truthfulness, consider an asymmetric joint strategy $F \neq G$. Then there exists $i$ such that $F_i \neq G_i$.
This reduces the expected payment by at least 
\begin{align}
\sum_{b=1}^m \Delta_b(i,i) S_b(F_i,G_i)
\end{align}

Since $F_i \neq G_i$, we have $\sum_{b=1}^m \Delta_b(F_i,G_i) < 0$ and there exists $l'$ such that $\Delta_{l'}(F_i,G_i) < 0$ (or $S_{l'}(F_i,G_i) = 0$). Therefore, the expected payment reduces by at least $\Delta_{l'}(i,i) > 0$.

Now consider symmetric, non-permutation strategy $F = G$. Then there exist $i \neq j$ such that $F_i = G_j = k$ and the expected payment includes 
\begin{align}
\sum_{b=1}^m \Delta_b(i,j) S_b(k,k) = \sum_{b=1}^m  \Delta_b(i,j)  < 0
\end{align}

The first equality uses the fact $S_b(k,k) = 1$ since $\Delta_b(k,k) > 0$ for each $b$.

\subsection{Proof of Theorem \ref{thm:detail-free-cah}}
\label{app:detail-free-cah}
We will write $\E{T,F,G}$ to denote the average expected score under strategies $F$ and $G$ when using the score matrix $T = \{T_b\}_{b=1}^m$. 
Suppose $S = \{S_b\}_{b=1}^m$ is the true scoring matrix and $\hat{S}=\{\hat{S}_b\}_{b=1}^m$ is the scoring matrix estimated from the data. Then
\begin{align}
\E{\hat{S},F,G} = \frac{1}{m}\sum_{b=1}^m \sum_{i,j} \Delta_b(i,j) \hat{S}_b(F_i,G_j) \leq \frac{1}{m} \sum_{b=1}^m \sum_{i,j : \Delta_b(i,j) > 0} \Delta_b(i,j) = \E{S,\eye,\eye}
\end{align}

Therefore, in order to show  $\E{\hat{S}, \eye, \eye} \geq \E{\hat{S},F,G} - \varepsilon$ it is enough to show that $\E{\hat{S}, \eye, \eye} \geq \E{{S},\eye,\eye} - \varepsilon$. Now
\begin{align}
&\abs{\frac{1}{m}\E{\hat{S},\eye,\eye} - \frac{1}{m}\E{S,\eye,\eye}} \nonumber \\
&= \abs{\frac{1}{m}\sum_{b=1}^m \sum_{i,j} \Delta_b(i,j) \left( \hat{S}_b(i,j) - S_b(i,j)\right)} 
= \abs{\frac{1}{m}\sum_{b=1}^m \sum_{i,j} \Delta_b(i,j) \left( \text{sign}\left(\hat{\Delta}_b(i,j)\right)  - \text{sign}\left( \Delta_b(i,j)\right)\right)} \nonumber \\
&\leq \frac{1}{m} \sum_{b=1}^m \sum_{i,j} \abs{\Delta_b(i,j) \left( \text{sign}\left(\hat{\Delta}_b(i,j)\right)  - \text{sign}\left( \Delta_b(i,j)\right)\right)} \leq \frac{1}{m} \sum_{b=1}^m \sum_{i,j} \abs{\hat{\Delta}_b(i,j) - \Delta_b(i,j)} \nonumber \\
&= \frac{1}{m} \sum_{b=1}^m \sum_{i,j} \abs{P_b(i,j) - T_b(i,j) - \frac{1}{(m-1)(m-2)} \sum_{\stackrel{t',t'' \in [m] \setminus \set{b}}{\& \ t' \neq t''}} \left( P_{t'}(i)P_{t''}(j) - T_{t'}(i)T_{t''}(j) \right) } \nonumber \\
& \leq \frac{1}{m} \sum_{b=1}^m \sum_{i,j} \abs{P_b(i,j) - T_b(i,j)} + \frac{1}{(m-1)(m-2)} \sum_{\stackrel{t',t'' \in [m] \setminus \set{b}}{\& \ t' \neq t''}} \abs{P_{t'}(i)P_{t''}(j) - T_{t'}(i)T_{t''}(j)} \nonumber \\
&= \frac{1}{m} \sum_{b=1}^m \sum_{i,j} \abs{P_b(i,j) - T_b(i,j)} +  \frac{1}{m(m-1)(m-2)} \sum_{b=1}^m \sum_{i,j} \sum_{\stackrel{t',t'' \in [m] \setminus \set{b}}{\& \ t' \neq t''}} \left|P_{t'}(i)\left( P_{t''}(j) - T_{t''}(j)\right)  \right. \nonumber \\
&+ \left. T_{t''}(j)\left( P_{t'}(i) - T_{t'}(i)\right) \right. \nonumber \\
&\leq \frac{1}{m}\sum_{b=1}^m \sum_{i,j} \abs{P_b(i,j) - T_b(i,j)} \\
&+ \frac{1}{m(m-1)(m-2)} \sum_{b=1}^m  \sum_{\stackrel{t',t'' \in [m] \setminus \set{b}}{\& \ t' \neq t''}} \left\{ \sum_j\abs{P_{t''}(j) - T_{t''}(j)}\sum_i P_{t'}(i) \right. 
+ \left.\sum_i \abs{P_{t'}(i) - T_{t'}(i)} \sum_j T_{t''}(j) \right\} \nonumber \\
&= \frac{1}{m}\sum_{b=1}^m \sum_{i,j} \abs{P_b(i,j) - T_b(i,j)}  \\
&+\frac{1}{m(m-1)(m-2)} \sum_{b=1}^m \sum_{\stackrel{t',t'' \in [m] \setminus \set{b}}{\& \ t' \neq t''}} \left\{ \sum_j\abs{P_{t''}(j) - T_{t''}(j)} 
+ \sum_i \abs{P_{t'}(i) - T_{t'}(i)}\right\}\label{eq:bound-partial}
\end{align}

Now if we have $\bigo{\frac{n^2}{\varepsilon^2}
	\log\left( \frac{m}{\delta}\right)}$ samples from
each joint distribution $P_b$ (where $n$ is the number
of signals) and $\bigo{\frac{n}{\varepsilon^2}
	\log\left( \frac{m}{\delta}\right)}$ from each
marginal distribution $P_b$, we can ensure that with
probability at least $1-\delta$, for all
$b=1,2,\ldots,m$ the following results hold (see
\cite{devroye2012combinatorial} for a proof)
\begin{equation}\label{eq:devroye}
\sum_{i,j} \abs{P_b(i,j) - T_b(i,j)} \leq \frac{\varepsilon}{3} \ \text{and}\ \sum_i \abs{P_b(i) - T_b(i)} \leq \frac{\varepsilon}{3}.
\end{equation}

Note: If we just had $O(n/\varepsilon^2 \log(1/\delta))$ samples for each
task, then we can guarantee~\eqref{eq:devroye}
for each task separately
with probability at least $1 - \delta$. By the union bound, this 
would give a success
probability of $1 - m \delta$ over all tasks. So in order to have
a $1 - \delta$
confidence bound, we need a $\log(m/\delta)$ factor in
the sample complexity.
%
Substituting the bounds from \cref{eq:devroye} in \cref{eq:bound-partial} and simplifying gives us $\abs{\E{\hat{S},\eye,\eye} - \E{S,\eye,\eye}} \leq \varepsilon$. Since there are $q$ agents providing reviews for each task, we get $q^2$ samples from each joint distribution and $q$ samples 
from each marginal distribution. So as long as $q =\bigom{ \frac{n}{\varepsilon^2} \log \left( \frac{m}{\delta}\right)}$ we have enough number of
samples and we are done.

\subsection{Additional Plots}
\label{app:uni-plots}
\begin{figure*}[t!]
	\minipage{0.24\textwidth}
	\includegraphics[width=\linewidth]{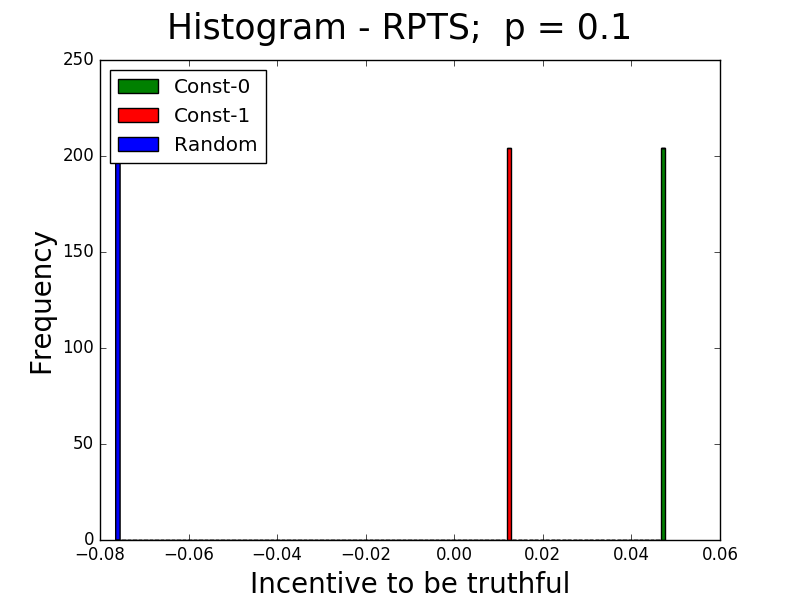}
	\endminipage\hfill
	\minipage{0.24\textwidth}
	\includegraphics[width=\linewidth]{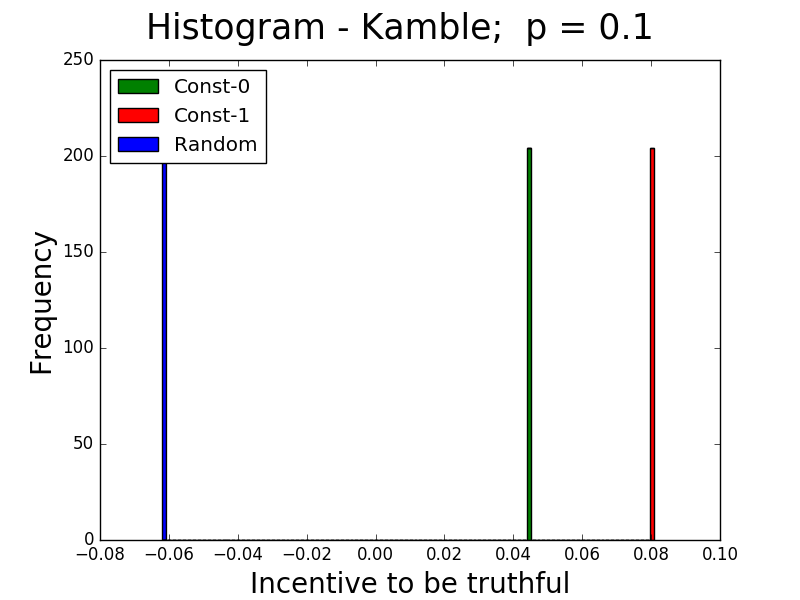}
	\endminipage\hfill
	\minipage{0.24\textwidth}%
	\includegraphics[width=\linewidth]{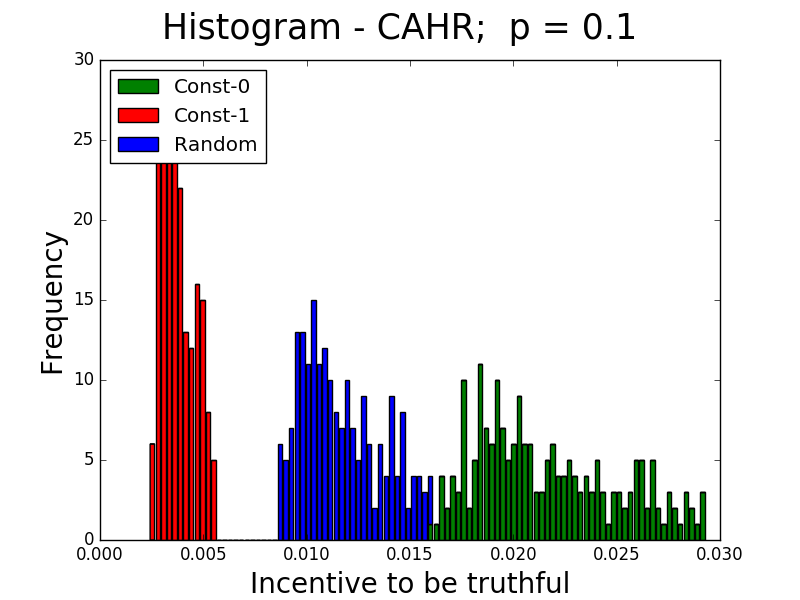}
	\endminipage\hfill
	\minipage{0.24\textwidth}
	\includegraphics[width=\linewidth]{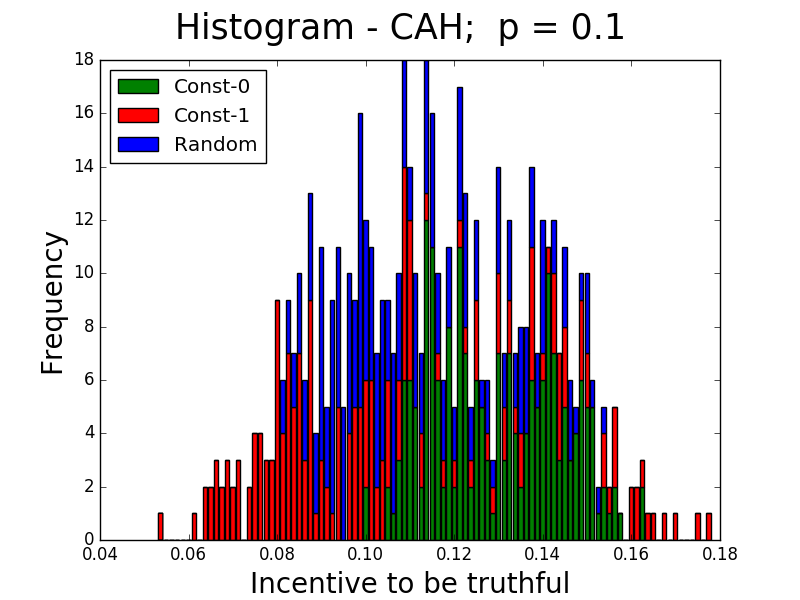}
	\endminipage\hfill
	\newline 
	\minipage{0.24\textwidth}
	\includegraphics[width=\linewidth]{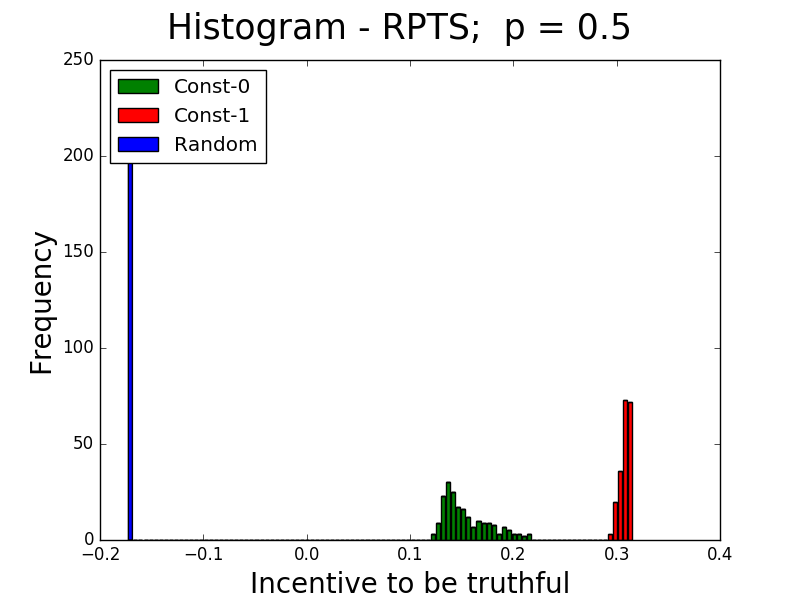}
	\endminipage\hfill
	\minipage{0.24\textwidth}
	\includegraphics[width=\linewidth]{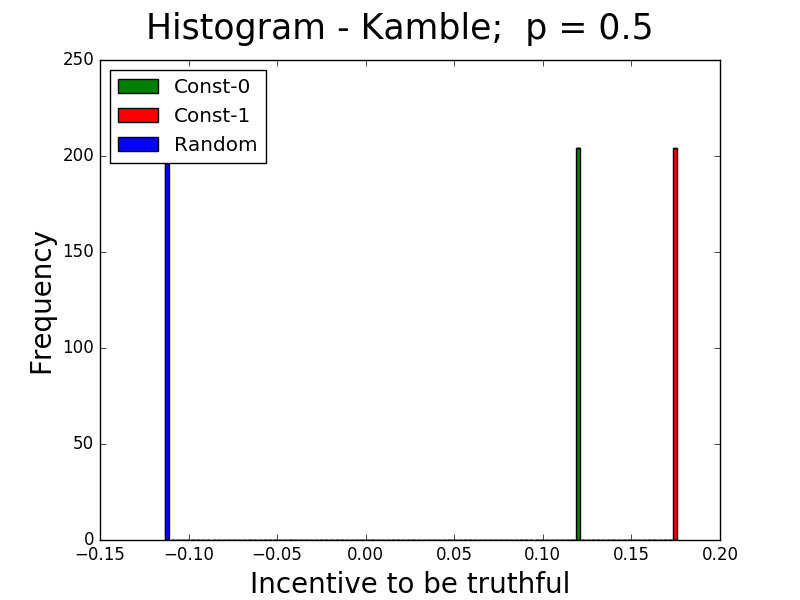}
	\endminipage\hfill
	\minipage{0.24\textwidth}%
	\includegraphics[width=\linewidth]{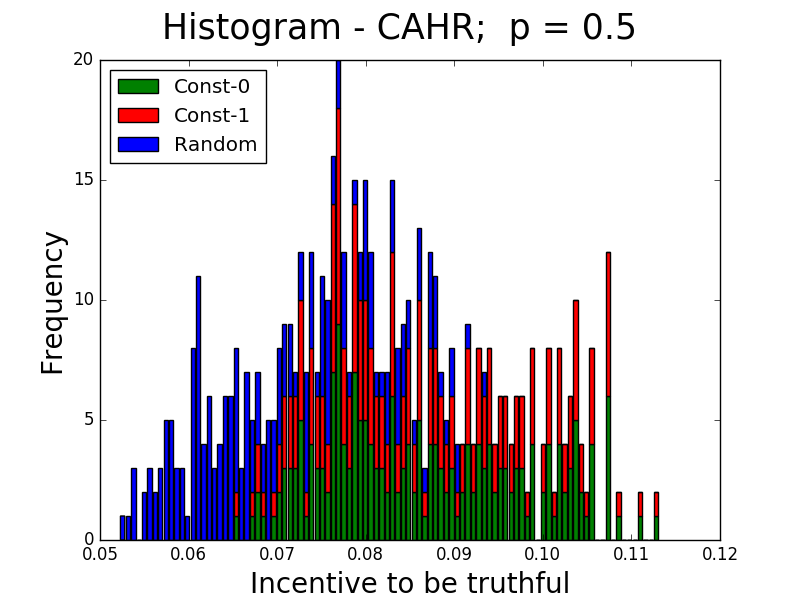}
	\endminipage\hfill
	\minipage{0.24\textwidth}
	\includegraphics[width=\linewidth]{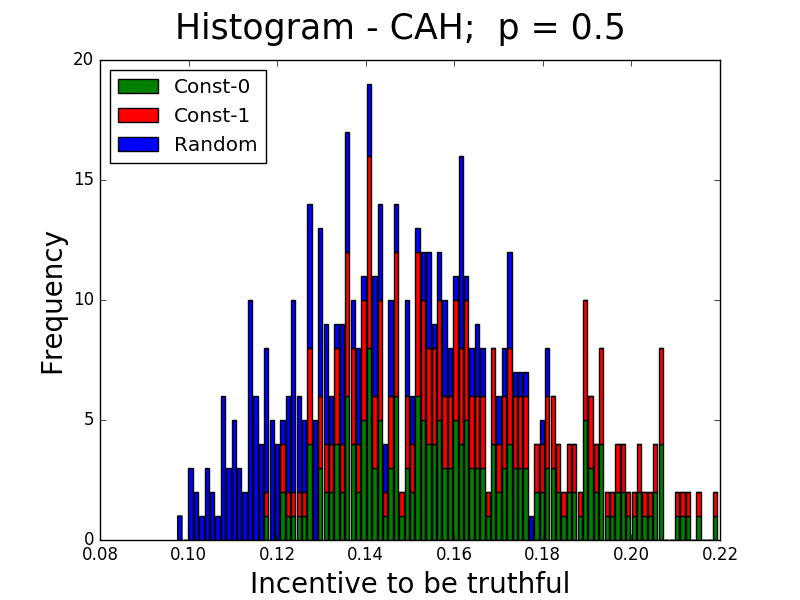}
	\endminipage\hfill
	\newline 
	\minipage{0.24\textwidth}
	\includegraphics[width=\linewidth]{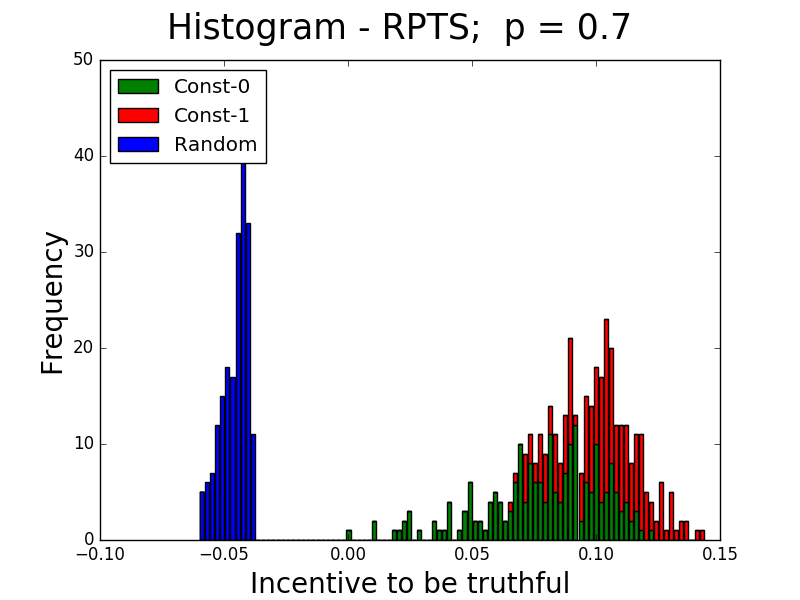}
	\endminipage\hfill
	\minipage{0.24\textwidth}
	\includegraphics[width=\linewidth]{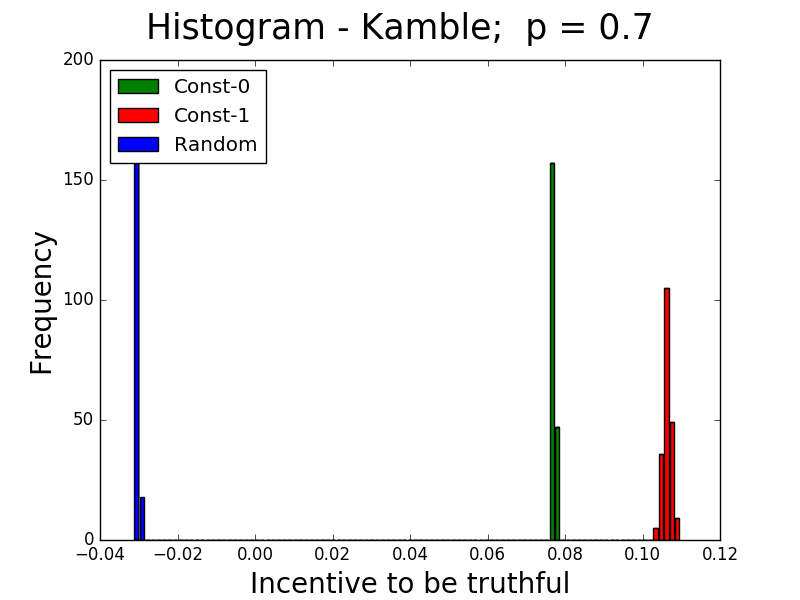}
	\endminipage\hfill
	\minipage{0.24\textwidth}%
	\includegraphics[width=\linewidth]{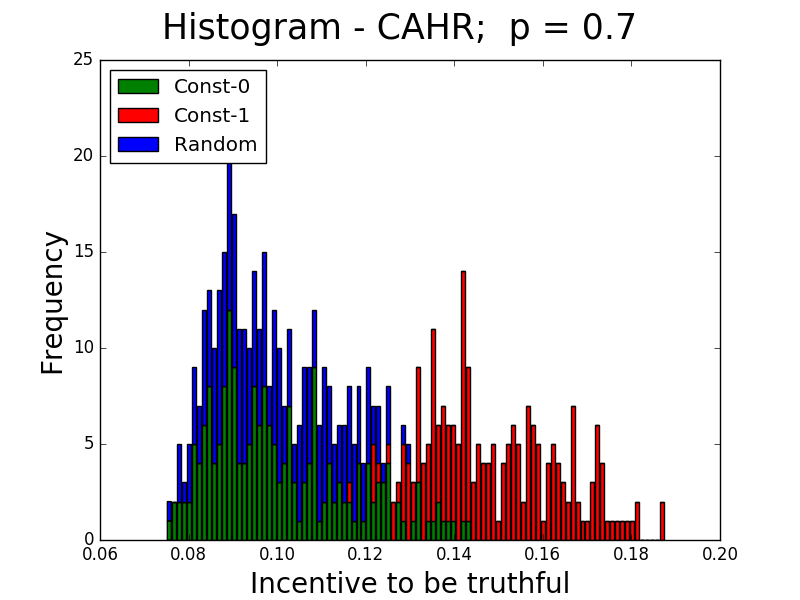}
	\endminipage\hfill
	\minipage{0.24\textwidth}
	\includegraphics[width=\linewidth]{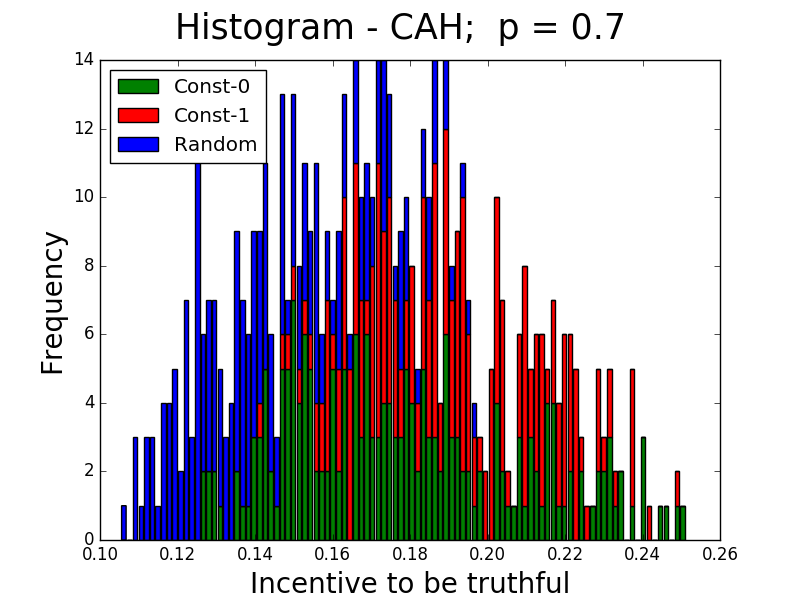}
	\endminipage\hfill
	\minipage{0.24\textwidth}
	\includegraphics[width=\linewidth]{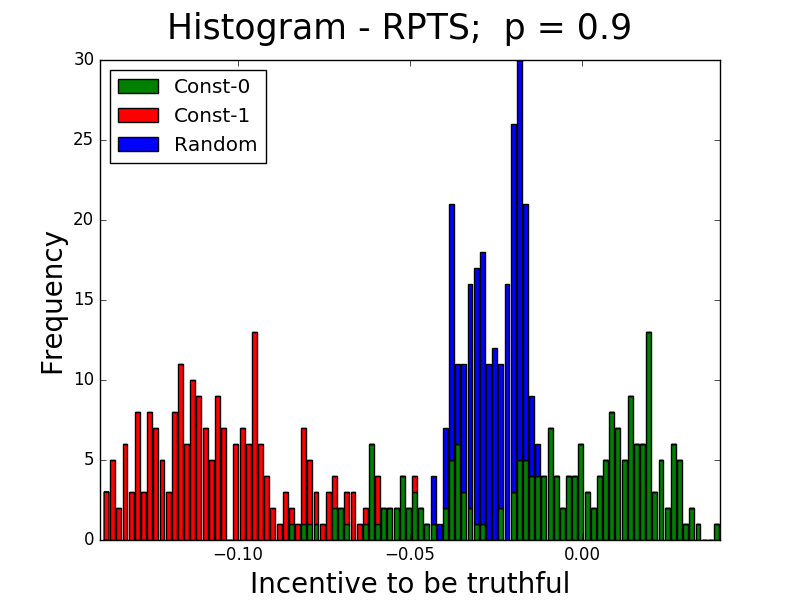}
	\endminipage\hfill
	\minipage{0.24\textwidth}
	\includegraphics[width=\linewidth]{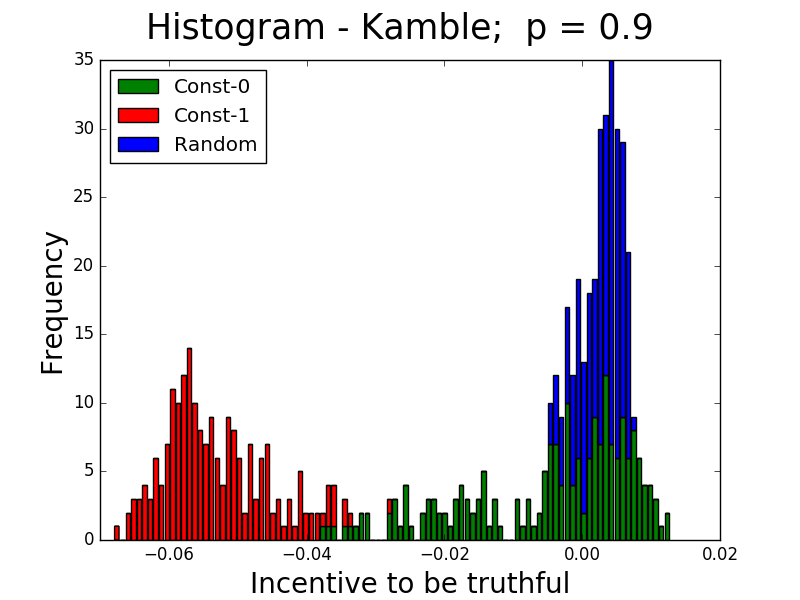}
	\endminipage\hfill
	\minipage{0.24\textwidth}%
	\includegraphics[width=\linewidth]{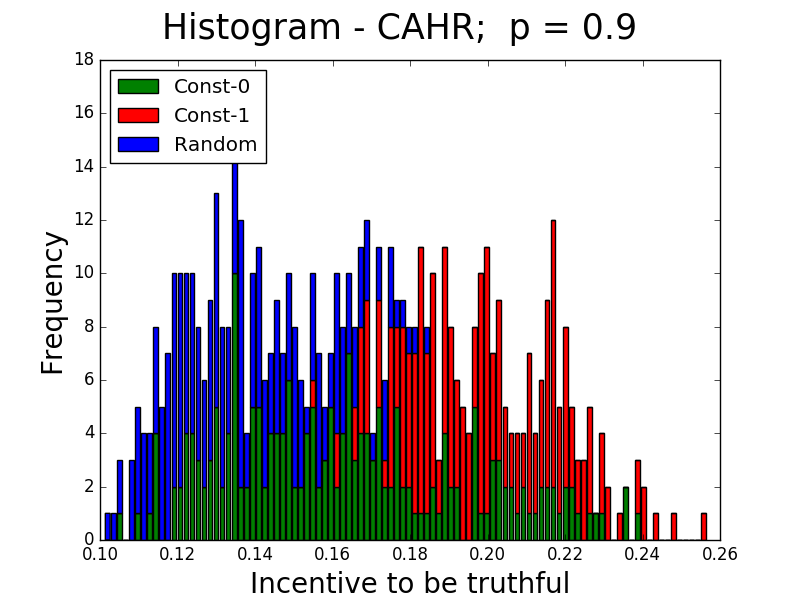}
	\endminipage\hfill
	\minipage{0.24\textwidth}
	\includegraphics[width=\linewidth]{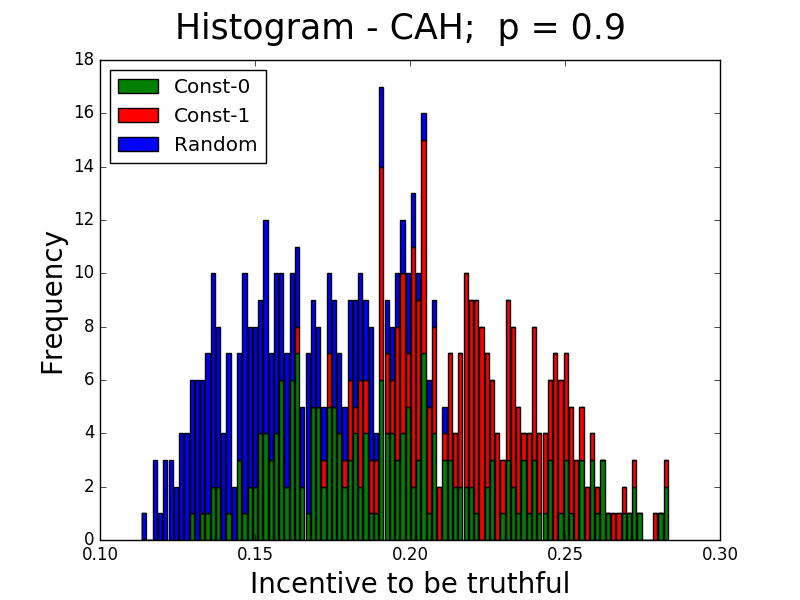}
	\endminipage\hfill
	\newline 
	\caption{Histograms for the 204
		(region, business type) pairs of
		expected benefit
		(averaged across questions) from truthful behavior
		vs.  some other strategy, when fraction $p$
		is truthful and fraction $1-p$ adopt the same,
		non-truthful
		strategy for $p = 0.1,0.5,0.7,0.9$  \label{fig:1}}
\end{figure*}
\end{document}